\documentclass[12pt, draftclsnofoot, onecolumn]{IEEEtran}
\usepackage{epsfig,graphicx,subfigure,psfrag,amsmath,cases,bm}
\usepackage{latexsym,amssymb,algorithm,mathtools}
\usepackage{algorithmic}
\usepackage{color}
\usepackage{url}
\usepackage{scrtime}
\usepackage{stfloats}
\usepackage{tablefootnote}
\usepackage{cite}
\usepackage{amsfonts,mathabx}
\usepackage{textcomp}
\usepackage{amsthm}
\usepackage{xcolor,capt-of}
\usepackage{multirow}
\makeatletter
\newcommand*{\rom}[1]{\expandafter\@slowromancap\romannumeral #1@}
\makeatother

\DeclareMathOperator{\mino}{minimize}
\newtheorem{remark}{Remark}
\newtheorem{lemma}{Lemma}

\allowdisplaybreaks
\title{Movable Antenna-Enhanced Multiuser Communication: Optimal Discrete Antenna Positioning and Beamforming}
\author{\IEEEauthorblockN {Yifei Wu\IEEEauthorrefmark{1}, Dongfang Xu\IEEEauthorrefmark{2}, Derrick Wing Kwan Ng\IEEEauthorrefmark{3}, Wolfgang Gerstacker\IEEEauthorrefmark{1}, and Robert Schober\IEEEauthorrefmark{1}}

\IEEEauthorrefmark {1}Friedrich-Alexander-Universit\"at
Erlangen-N\"urnberg, Germany\\
\IEEEauthorrefmark {2}The Hong Kong University of Science and Technology, Hong Kong\\
\IEEEauthorrefmark {3}The University
of New South Wales, Australia%
}
\begin{document}
\maketitle
\begin{abstract}
    Movable antennas (MAs) are a promising paradigm to enhance the spatial degrees of freedom of conventional multi-antenna systems by flexibly adapting the positions of the antenna elements within a given transmit area. In this paper, we model the motion of the MA elements as discrete movements and study the corresponding resource allocation problem for MA-enabled multiuser multiple-input single-output (MISO) communication systems. Specifically, we jointly optimize the beamforming and the MA positions at the base station (BS) for the minimization of the total transmit power while guaranteeing the minimum required signal-to-interference-plus-noise ratio (SINR) of each individual user. To obtain the globally optimal solution to the formulated resource allocation problem, we develop an iterative algorithm capitalizing on the generalized Bender's decomposition with guaranteed convergence. Our numerical results demonstrate that the proposed MA-enabled communication system can significantly reduce the BS transmit power and the number of antenna elements needed to achieve a desired performance compared to state-of-the-art techniques, such as antenna selection. Furthermore, we observe that refining the step size of the MA motion driver improves performance at the expense of a higher computational complexity.
\end{abstract}
\section{Introduction}
Multiple-input multiple-output (MIMO) transmission is widely envisioned as a key technique in fulfilling the tremendous data traffic demands in sixth-generation (6G) wireless networks. By utilizing multiple antennas, MIMO can effectively leverage the spatial resources of wireless channels to deliver significant performance improvements, including improved data transmission rates\cite{mietzner2009multiple}, enhanced physical layer security\cite{tsai2014power}, and realize new paradigms such as integrated sensing and communication\cite{xu2022robust}, etc. However, conventional MIMO systems require multiple parallel radio frequency (RF) chains, leading to high hardware cost and computational complexity \cite{mietzner2009multiple}. %Antenna selection (AS) is a practical approach for reducing the number of RF chains and achieving a large portion of capacity in MIMO systems by selecting a small subset of antennas with favorable channels from a large set of candidate antennas. 
To mitigate the cost and complexity introduced by a large number of RF chains, antenna selection (AS) has been advocated as a practical approach for the realization of a MIMO system. Indeed, the goal of AS is to capture the possible diversity and multiplexing gains of a given MIMO system by selecting a small subset of antennas with favorable channel characteristics from a large set of candidate antennas, thereby reducing the required number of RF chains \cite{sanayei2004antenna}. However, %increasing the number of candidate antennas may lead to higher hardware costs. Moreover, 
in conventional MIMO systems with or without AS, the antennas are deployed at fixed positions. As such, the inherent variations of the channel across the spatial continuous transmitter area cannot be fully exploited, which limits system performance.%deploy antenna elements at only fixed positions with decoupling spacing. As such, the inherent channel variation in the continuous spatial field cannot be fully exploited. In other words, the spatial degree of freedom (DoF) of the continuous spatial field cannot be fully utilized by the fixed-position antennas \cite{ma2022mimo}\cite{zhu2023movable}. 
%In order to mitigate the cost and complexity of employing a large number of RF chains, the antenna selection (AS) technique is advocated for MIMO communications. The goal of AS is to retain the spatial diversity of the MIMO system by selecting a small subset of antennas with the most favorable channel characteristics, thereby reducing the number of RF chains required. However, AS requires still a large number of candidate antenna elements to guarantee a promising performance, leading to a high hardware cost and increasing complexity in channel estimation. Furthermore,

To fully exploit the spatial variation of wireless channels within a given spatial transmitter area, the emerging holographic MIMO technique has been proposed in the literature \cite{huang2020holographic}. In particular, holographic MIMO surfaces consist of numerous miniature passive elements, spaced at sub-wavelength distances, which can be electronically controlled to manipulate the electromagnetic properties of the transmitted or reflected waves \cite{xu2020resource}. In practice, by utilizing zero-spacing continuous antenna elements, the available spatial degrees of freedom (DoFs) of the spatially continuous transmitter area can be fully leveraged by holographic MIMO. However, the large number of antenna elements required for holographic MIMO presents a critical challenge for both channel estimation and data processing, which hinders its practical implementation \cite{huang2020holographic}. %Thus, the holographic MIMO remains primarily in the stage of theoretical exploration.

Inspired by the spatial DoFs facilitated by holographic MIMO surfaces, a new MIMO concept based on movable antennas (MAs) has been proposed as a bridge technology between holographic MIMO and conventional MIMO \cite{zhu2022modeling}. In MA-enabled systems, each antenna element is connected to a radio frequency (RF) chain via a flexible cable and its physical position can be adjusted within a designated spatial region exploiting some electromechanical device, such as a stepper motor \cite{ma2022mimo}. This mobility allows the repositioning of the MA element at an optimal location for establishing favorable spatial antenna correlations so as to maximize the capacity of the MIMO system. In contrast to conventional MIMO systems, which comprise a set of antenna elements mounted at fixed locations, MA systems can utilize the full spatial DoFs within the available spatial transmitter area by leveraging the flexible movement of the MAs. Moreover, since MA systems require only a small number of antenna elements to exploit the available DoFs, the computational complexity entailed by the required signal processing is significantly reduced compared to holographic MIMO systems \cite{zhu2022modeling,ma2022mimo,zhu2023movable}.

To fully unleash the potential of the MA-enabled systems, a few initial works explored the joint design of beamforming and antenna positioning. For instance, in \cite{ma2022mimo}, a suboptimal algorithm based on alternating optimization (AO) was proposed for MA-enabled MIMO systems, where both the base station (BS) and multiple users were equipped with MAs. Also, the authors in \cite{zhu2023movable} considered a multiuser MA-enabled uplink communication system comprising multiple single-MA users and a BS equipped with a fixed antenna array. The BS employed zero-forcing (ZF) or minimum mean square error (MMSE) combining and the MA positions were adjusted using a gradient descent (GD) method. However, both \cite{ma2022mimo} and \cite{zhu2023movable} assume optimistically that the positions of MA elements can be adjusted freely within a given region, which may not be practical. In the prototype designs of MA-enabled systems shown in \cite{zhuravlev2015experimental} and \cite{basbug2017design}, the motion control of the employed electromechanical devices is discrete with finite precision. Thus, the transmitter area is quantized in a pace\cite{basbug2017design}, leading to a finite spatial resolution instead of the infinite resolution assumed in \cite{zhu2022modeling,ma2022mimo,zhu2023movable}. Moreover, the AO-based algorithm in \cite{ma2022mimo} and the GD-based method in \cite{zhu2023movable} cannot guarantee the joint optimality of the BS beamformer and the MA positions, as their performance highly relies on the selection of the initial point. Thus, in this paper, we investigate for the first time the jointly globally optimal design of the BS beamforming matrix and MA positions for a multiuser MA-enabled downlink system with a spatially discrete transmitter area to fully reveal the potential of MA-enabled systems. The main contributions of this paper can be summarized as follows:
\begin{itemize}
\item Taking into account the discrete nature of state-of-the-art electromechanical control hardware, we introduce a novel MA position model that facilitates the position optimization of the MA elements.
\item To obtain the jointly globally optimal MA positions and BS beamforming matrix, we propose a series of mathematical transformations that allow us to recast the considered challenging resource allocation problem into a tractable mixed integer nonlinear programming (MINLP) problem.
\item We propose an iterative algorithm exploiting the generalized Bender's decomposition (GBD) to obtain the globally optimal solution of the considered joint design problem. The proposed algorithm can serve as a performance benchmark for any suboptimal design, e.g., those in \cite{ma2022mimo}, \cite{zhu2023movable}.
\end{itemize}
The remainder of this paper is organized as follows: In Section \rom{2}, we introduce the system model for the considered MA-enabled multiuser multiple-input single-output (MISO) communication system with a spatially discrete transmitter area and formulate the corresponding resource allocation problem. In Section \rom{3}, the globally optimal solution for the MA positions and the BS beamforming matrix is provided. Section \rom{4} evaluates the performance of the proposed optimal design via numerical simulations, and Section \rom{5} concludes this paper.

\textit{Notation:} 
Vectors and matrices are denoted by boldface lower case and boldface capital letters, respectively. $\mathbb{R}^{N\times M}$ and $\mathbb{C}^{N\times M}$ represent the space of $N\times M$ real-valued and complex-valued matrices, respectively. $|\cdot|$ and $||\cdot||_2$ stand for the absolute value of a complex scalar and the $l_2$-norm of a vector, respectively. $(\cdot)^T$, $(\cdot)^*$, and $(\cdot)^H$ denote the transpose, the conjugate, and the conjugate transpose of their arguments, respectively. $\mathbf{I}_{N}$ refers to the identity matrix of dimension $N$. $\mathrm{Tr}(\cdot)$ is the trace of the input argument. $\mathbf{0}_{L}$ and $\mathbf{1}_L$ represent the all-zeros and all-ones vector of length $L$, respectively. $\mathbf{A}\succeq\mathbf{0}$ indicates that $\mathbf{A}$ is a positive semidefinite matrix. $\mathrm{diag}(\mathbf{a})$ denotes a diagonal matrix whose main diagonal elements are given by the entries of vector $\mathbf{a}$. 
$\mathrm{Re}\{\cdot\}$ and $\mathrm{Im}\{\cdot\}$ represent the real and imaginary parts of a complex number, respectively. $\mathbb{E}[\cdot]$ refers to statistical expectation.
\section{MA-Enhanced Multiuser System Model}

\begin{figure}
    \centering
    \includegraphics[width=3.8in]{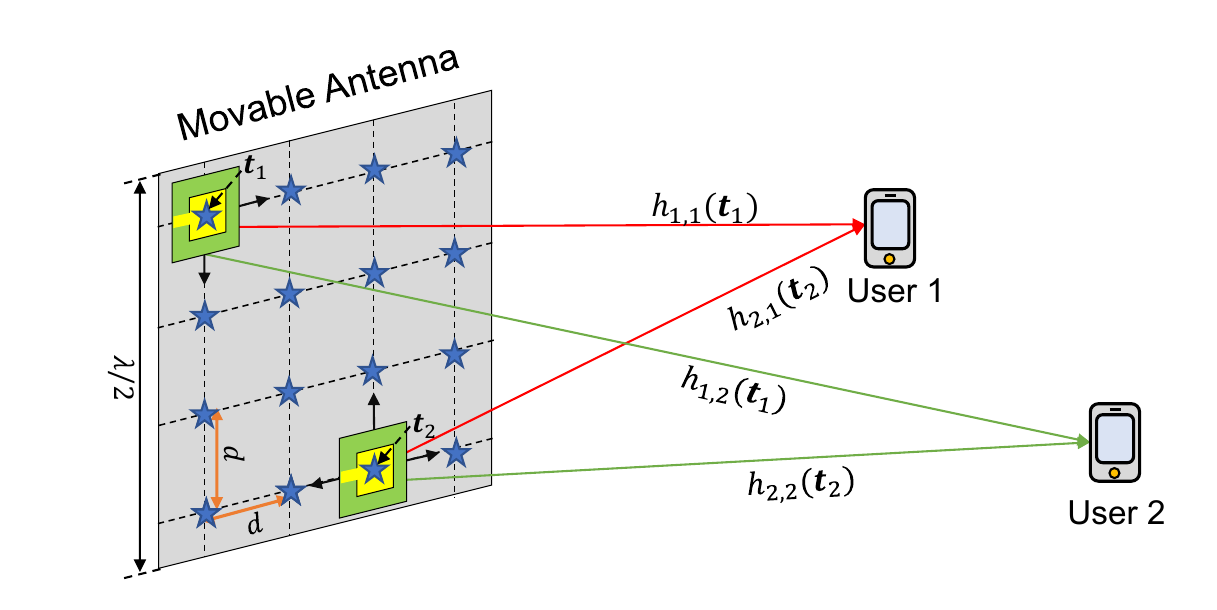}
    \caption{Transmission from $M=2$ movable antenna elements with $N=16$ possible discrete positions to $K=2$ users ($\bigstar$ symbols represent feasible antenna positions).}
    \label{fig:MA_system_model}
\end{figure}
\subsection{Channel Model}
We consider a multiuser wireless communication system comprising a BS and $K$ users. The BS is equipped with $M$ MA elements for serving $K$ single-antenna users. %The $M$ MA elements are assumed to be the same, i.e., the channels between different MA elements to the users in the same position are the same. 
The positions of the MA elements can be adjusted simultaneously within a given two-dimensional transmitter area. To investigate the maximal possible performance of the considered MA-enabled system, we assume that perfect channel state information (CSI) with respect to the transmitter area is available at the BS \cite{ma2022mimo}, \cite{zhu2023movable} \footnote{The robust design of MA-enabled systems taking into account imperfect CSI is an interesting topic for future work.}. Since practical electromechanical devices can only provide a horizontal or vertical movement by a fixed increment $d$ in each step \cite{zhuravlev2015experimental}\cite{basbug2017design}, the transmitter area of the MA-enabled communication system is quantized \cite{basbug2017design}\footnote{The value of step size $d$ depends on the precision of the employed electromechanical devices and may vary in different MA-enabled systems. }. We collect the $N$ possible discrete positions of the MAs in set $\mathcal{P}=\{\mathbf{p}_1,\cdots, \mathbf{p}_N\}$, where the distance between the neighboring positions is equal to $d$ in horizontal or vertical direction, as shown in Fig. \ref{fig:MA_system_model}. Here, $\mathbf{p}_n=[x_n,y_n]$ represents the $n$-th candidate position with horizontal coordinate $x_n$ and vertical coordinate $y_n$. 
In other words, the feasible set of the position of the $m$-th MA element, $\mathbf{t}_m$, is given by $\mathcal{P}$, i.e., $\mathbf{t}_m\in\mathcal{P}$. In the considered MA-enabled MIMO system, the physical channel can be reconfigured by adjusting the positions of the MA elements. The channel vector between the $m$-th MA element and the $K$ users is denoted by $\mathbf{h}_m(\mathbf{t}_m)=[h_{m,1}(\mathbf{t}_m),\cdots,h_{m,K}(\mathbf{t}_m)]^T$ and depends on the position of the $m$-th MA element $\mathbf{t}_m$, where $h_{m,k}(\mathbf{t}_m)\in\mathbb{C}$ denotes the channel coefficient between the $m$-th MA element and the $k$-th user.
%Since the antenna positions are selected from a discrete feasible set $\mathcal{P}$, the MA position optimization problem is equivalent to selecting $M$ antennas from a holographic MIMO surface, where all $N$ discrete positions in $\mathcal{P}$ are equipped with antenna elements. The position of the $m$-th MA element $\mathbf{t}_m$ can  in the Hence, the 
%Then, we define the binary position selection vector of the $m$-th MA element $\mathbf{b}_m=\big[b_m[1],\cdots,b_m[N]\big]^T,\hspace*{1mm}\forall m$, where $b_m[n]\in\left\{0,\hspace*{1mm}1\right\}$, $\sum_{n=1}^{N}b_m[n]=1$, $\forall m$. Note that $b_m[n']=1$ if and only if the $n'$-th discrete position in $\mathcal{P}$ is selected for the $m$-th MA element.
%Let $\mathbf{h}_{m}(\mathbf{b}_m)=[{h}_{m,1}(\mathbf{b}_m),\cdots,{h}_{m,K}(\mathbf{b}_m)]\in \mathbb{C}^{K\times 1}$ denote the channel vector between the $m$-th MA element and all $K$ users, where ${h}_{m,k}(\mathbf{b}_m)\in \mathbb{C}$ represents the channel coefficient from the $m$-th MA element to the $k$-th user depending on the position selection vector of the $m$-th MA element $\mathbf{b}_m$. %In this work, we assume that the MA system is equipped with $M$ MA elements with the same property, i.e., $h_{m,k}(\mathbf{b}_m)=h_{m',k}(\mathbf{b}_m), \forall m\neq m'$. 
Next, we define a matrix $\hat{\mathbf{H}}_{m}=[\mathbf{h}_{m}(\mathbf{p}_1),\cdots,\mathbf{h}_{m}(\mathbf{p}_N)]\in\mathbb{C}^{K\times N}$ to collect the channel vectors from the $m$-th MA element to all $K$ users for all $N$ feasible discrete MA locations. Then, $\mathbf{h}_{m}(\mathbf{t}_m)$ can be expressed as
\begin{equation}
    \mathbf{h}_{m}(\mathbf{t}_m)=\hat{\mathbf{H}}_{m}\mathbf{b}_m,
\end{equation}
where $\mathbf{b}_m=\big[b_m[1],\cdots,b_m[N]\big]^T$. Here, $b_m[n]\in\left\{0,\hspace*{1mm}1\right\}$ with $\sum_{n=1}^{N}b_m[n]=1$ is a binary variable defining the position of the $m$-th MA element. %Note that $\hat{\mathbf{H}}_{m}$ can be explained as the MIMO channel matrix from an effective holographic MIMO surface to all $K$ users, where all $N$ positions in $\mathcal{P}$ are equipped with antenna elements. Thus, in the following, we refer to $\hat{\mathbf{H}}_{m}$ as the effective holographic MIMO channel for the $m$-th MA element. 
For the considered MA-enabled multiuser MISO system, the channel matrix between the BS and the $K$ users, $\mathbf{H}=[\mathbf{h}_{1}(\mathbf{t}_1),\cdots,\mathbf{h}_{M}(\mathbf{t}_M)]\in\mathbb{C}^{K\times M}$, is then given by
\begin{equation}
    \mathbf{H}=\hat{\mathbf{H}}\mathbf{B},
\end{equation}
where matrices $\hat{\mathbf{H}}\in \mathbb{C}^{K\times MN}$ and $\mathbf{B}\in \mathbb{C}^{MN\times M}$ are defined as follows, respectively,
\begin{eqnarray}
\hat{\mathbf{H}}&\hspace*{-2mm}=\hspace*{-2mm}&
[\hat{\mathbf{H}}_{1},\cdots,\hat{\mathbf{H}}_{M}],\\
\mathbf{B}&\hspace*{-2mm}=\hspace*{-2mm}&
  \begin{bmatrix}
    \mathbf{b}_1 & \mathbf{0}_{N} & \mathbf{0}_{N} & \cdots & \mathbf{0}_{N}\\
    \mathbf{0}_{N} & \mathbf{b}_2 & \mathbf{0}_{N} &\cdots & \mathbf{0}_{N}\\
    \ldots & \ldots & \ldots & \ldots & \ldots\\
    \mathbf{0}_{N} & \mathbf{0}_{N} & \mathbf{0}_{N} & \hspace*{1mm}\cdots & \mathbf{b}_M
  \end{bmatrix}.
\end{eqnarray}
%Since one discrete position can only be selected by one MA element, the selection vectors of different MA elements cannot be the same, i.e., $\mathbf{b}_m\neq \mathbf{b}_{m'},\ \forall m\neq m'$, which is equivalent to $\sum_{m=1}^Mb_m[n]\leq 1,\hspace*{1mm}\forall n$. 
Next, we define $\hat{\mathbf{h}}_k\in\mathbb{C}^{1\times MN}$ as the $k$-th row of $\hat{\mathbf{H}}$. Then, the received signal of the $k$-th user $y_k$ is given by
\begin{equation}
    y_k=\hat{\mathbf{h}}_k\mathbf{B}\mathbf{W}\mathbf{s}+n_k,
\end{equation}
where $\mathbf{s}=[s_1,\cdots,s_K]^T\in\mathbb{C}^{K\times 1}$ represents the information-carrying symbol vector transmitted to the users. Here, $s_j\in\mathbb{C}$ denotes the symbol transmitted to the $j$-th user and $\mathbb{E}[|s_j|^2]=1$, $\mathbb{E}[s_j^*s_i]=0$, $j\neq i$, $\forall j,i\in\{1,\cdots,K\}$. $\mathbf{W}=[\mathbf{w}_1,\cdots,\mathbf{w}_K]$ denotes the linear beamforming matrix at the BS, where $\mathbf{w}_k\in \mathbb{C}^{M\times 1}$ represents the linear beamforming vector for the $k$-th user. $n_k\in\mathbb{C}$ stands for the additive white Gaussian noise at the $k$-th user with zero mean and variance $\sigma_k^2$. For notational simplicity, we define sets $\mathcal{K}\in\{1,\cdots,K\}$, $\mathcal{M}\in\{1,\cdots,M\}$, and $\mathcal{N}\in\{1,\cdots,N\}$ to collect the indices of the users, MA elements, and candidate positions of the MA elements, respectively.
\subsection{Resource Allocation Problem}
 By introducing an auxiliary matrix $\mathbf{X}=\mathbf{B}\mathbf{W},\ \mathbf{X}\in\mathbb{C}^{MN\times K}$, the received signal of the $k$-th user can be rewritten as follows
\begin{equation}
    y_k=\hat{\mathbf{h}}_k\mathbf{X}\mathbf{s}+n_k.
\end{equation}
Thus, the signal-to-interference-plus-noise ratio (SINR) of the $k$-th user is given by
\begin{equation}
    \mathrm{SINR}_k=\frac{|\hat{\mathbf{h}}_k^H\mathbf{x}_k|^2}{\sum_{k'\in\mathcal{K}\setminus\{k\}}|\hat{\mathbf{h}}_k^H\mathbf{x}_{k'}|^2+\sigma_{k}^2},
\end{equation}
where $\mathbf{x}_k$ denotes the $k$-th column of $\mathbf{X}$. Due to the limitation of antenna size, two MA elements cannot be placed arbitrarily close to each other. Thus, the center-to-center distance between any pair of MA elements must be greater than a minimum distance $D_{\mathrm{min}}$. We define distance matrix $\mathbf{D}\in\mathbb{C}^{N\times N}$, where element $D_{n,n'}$ in the $n$-th row and $n'$-th column of $\mathbf{D}$ denotes the distance between the $n$-th candidate position and the $n'$-th candidate position in $\mathcal{P}$. Thus, the minimum distance constraint between any pair of MA elements can be formulated as 
\begin{equation}
    \mathbf{b}_m^T\mathbf{D}\mathbf{b}_{m'}\geq D_{\mathrm{min}},\ m\neq m',\ \forall m, m'\in\mathcal{M}.
\end{equation}
In this paper, we aim to minimize the BS transmit power while guaranteeing a minimum required SINR for each user. The resulting resource allocation problem can be formulated as%\footnote{The formulated problem can be regarded as a generalized AS problem \cite{sanayei2004antenna} with an additional minimum distance constraint C5. Thus, the proposed method can serve as an optimal algorithm to solve the AS problem in \cite{sanayei2004antenna}. Note that an optimal algorithm for the AS problems is not available, yet. }
\begin{eqnarray}
\label{Ori_Problem}
    %&&\hspace*{-4mm}\underset{\mathbf{X},\mathbf{W},\bm{\Phi}}{\mino}\hspace*{2mm}\min_{k\in\mathcal{K}}\hspace*{4mm}\mathrm{SINR}_k\notag\\
    &&\hspace*{-4mm}\underset{\mathbf{X},\mathbf{W},\mathbf{B}}{\mino}\hspace*{2mm}\sum_{k\in\mathcal{K}}\left\|\mathbf{w}_k\right\|_2^2\notag\\
    %&&\hspace*{1mm}\mbox{s.t.}\hspace*{7mm} \mbox{C1:}\hspace*{1mm}\sum_{k\in\mathcal{K}}\left\|\mathbf{w}_k\right\|_2^2\leq P,\notag\\
    &&\hspace*{2mm}\mbox{s.t.}\hspace*{7mm} \mbox{C1:}\hspace*{1mm} \frac{|\hat{\mathbf{h}}_k^H\mathbf{x}_k|^2}{\sum_{k'\in\mathcal{K}\setminus\{k\}}|\hat{\mathbf{h}}_k^H\mathbf{x}_{k'}|^2+\sigma_{k}^2}\geq \gamma_{k},\hspace*{1mm}\forall k\in \mathcal{K},\notag\\
    &&\hspace*{14mm}\mbox{C2:}\hspace*{1mm}\mathbf{X}=\mathbf{B}\mathbf{W},\notag\\
    &&\hspace*{14mm}\mbox{C3:}\hspace*{1mm} b_m[n]\in \{0,1\},\hspace*{1mm} \forall n\in \mathcal{N}, \forall m\in \mathcal{M},\notag\\
    &&\hspace*{14mm}\mbox{C4:}\hspace*{1mm} \sum_{n=1}^Nb_m[n]=1,\hspace*{1mm}\forall m\in \mathcal{M},\notag\\
    %&&\hspace*{14mm}\mbox{C5:}\hspace*{1mm} \sum_{m=1}^Mb_m[n]\leq 1,\hspace*{1mm}\forall n\in \mathcal{N},\notag\\
    &&\hspace*{14mm}\mbox{C5:}\hspace*{1mm} \mathbf{b}_m^T\mathbf{D}\mathbf{b}_{m'}\geq D_{\mathrm{min}},\ m\neq m',\ \forall m, m'\in\mathcal{M}.
    %&&\hspace*{14mm}\mbox{C3:}\hspace*{1mm}\theta_n\in\left\{0,\Delta\theta,\cdots,(L-1)\Delta\theta\right\}, \hspace*{1mm} \forall n\in \mathcal{N}.\vspace*{-5mm}
\end{eqnarray}
Optimization problem \eqref{Ori_Problem} is nonconvex due to bilinear constraint C2, binary constraint C3, and binary quadratic constraint C5. Therefore, problem \eqref{Ori_Problem} is an NP-hard combinatorial problem. %Inspired by the optimal resource allocation design in \cite{wu2023globally}, %{\color{red}(I will also cite this in other papers)}, 
In the next section, we develop a GBD-based iterative algorithm to obtain the global optimum of \eqref{Ori_Problem}.

\section{Solution of Optimization Problem}
In this section, we leverage the GBD method in \cite{geoffrion1972generalized} to obtain the globally optimal solution to \eqref{Ori_Problem}. In the following, we first transform \eqref{Ori_Problem} into an equivalent MINLP problem, which provides a foundation for the development of the proposed GBD-based optimal algorithm.
\subsection{Problem Reformulation}
The globally optimal solution of \eqref{Ori_Problem} cannot be obtained by the GBD approach directly due to coupled constraint C2 \cite{geoffrion1972generalized} \cite{wu2023globally}. Thus, we present Lemma 1 to reformulate constraint C2 to facilitate the application of the GBD approach.
\begin{lemma}
Equality constraint C2 is equivalent to the following linear matrix inequality (LMI) constraints, 
\begin{eqnarray}
\mathrm{{C2a:}}&\hspace*{1mm}\label{sdp}
   \begin{bmatrix}
        \mathbf{U} & \mathbf{X} & \mathbf{B}\\
        \mathbf{X}^H & \mathbf{V} & \mathbf{W}^H\\
        \mathbf{B}^H & \mathbf{W} & \mathbf{I}_K
    \end{bmatrix}&\succeq \mathbf{0},\\
{\mathrm{C2b}}\mbox{:}&\hspace*{1mm}\label{DC}\mathrm{Tr}\big(\mathbf{U}\big)-M\leq 0,\vspace*{-2mm}
%    \mathrm{Tr}\left(\mathbf{U}-\mathbf{B}\mathbf{B}^H\right)&\leq0,\vspace*{-2mm}
\end{eqnarray}
where $\mathbf{U}\in\mathbb{C}^{N\times N}$ and $\mathbf{V}\in\mathbb{C}^{K\times K}$ are two auxiliary optimization variables with $\mathbf{U}\succeq \mathbf{0}$ and $\mathbf{V}\succeq \mathbf{0}$.
\end{lemma}
\begin{proof}
Based on \cite[Lemma 1]{6698281}, equality constraint C2 is equivalent to LMI constraint C2a and inequality constraint:
\begin{eqnarray}\label{DC2}
\overline{\mbox{C2b}}\mbox{:}&%\hspace*{1mm}\label{DC}\mathrm{Tr}\big(\mathbf{U}\big)-M\leq 0,\vspace*{-2mm}
    \mathrm{Tr}\left(\mathbf{U}-\mathbf{B}\mathbf{B}^H\right)&\leq0.\vspace*{-2mm}
\end{eqnarray}
The left-hand side of \eqref{DC2} can be rewritten as
\begin{equation}
    %\begin{aligned}
         \mathrm{Tr}\left(\mathbf{U}-\mathbf{B}\mathbf{B}^H\right)\overset{(a)}{=}\mathrm{Tr}\left(\mathbf{U}\right)-\sum_{m=1}^{M}\mathrm{Tr}\left(\mathbf{b}_m\mathbf{b}_m^H\right)\overset{(b)}{=}
         \mathrm{Tr}\left(\mathbf{U}\right)-M,
\end{equation}
where the above equalities (a) and (b) hold due to the additivity of the matrix trace and the definition of binary decision vector $\mathbf{b}_m,\forall m$, respectively. Thus, inequalities C2b and $\overline{\mbox{C2b}}$ are equivalent, which completes our proof.
\end{proof}
Note that C2a and ${\mbox{C2b}}$ are both convex constraints. On the other hand, the minimum distance constraint C5 in \eqref{Ori_Problem} is still non-convex. Here, we reformulate the quadratic inequality constraint C5 into three linear inequality constraints by exploiting the following lemma, the proof of which can be found in \cite{glover1974converting}.
\begin{lemma}
    The inequality constraint C5 is equivalent to the following linear inequality constraints% by introducing a binary auxiliary vector $\mathbf{y}=[y_{1,2,1,1},\cdots,y_{m,m',i,j},\cdots,y_{M-1,M,N,N}]$, $m\neq m'$, $\forall m,m'\in\mathcal{M}$ and $\forall i,j\in\mathcal{N}$:
\begin{eqnarray}
    &&\mathrm{C5a:}\hspace*{2mm}\sum_{i\in\mathcal{N}}\sum_{j\in\mathcal{N}}D_{i,j}y_{m,m',i,j}\geq D_{\mathrm{min}},\hspace*{2mm}m\neq m',\hspace*{2mm}\forall m,m'\in\mathcal{M},\notag\\
    &&\mathrm{C5b:}\hspace*{2mm}y_{m,m',i,j}\leq \min\left\{b_m[i],b_m[j]\right\},\hspace*{2mm}m\neq m',\hspace*{2mm} \forall m,m'\in\mathcal{M},\hspace*{2mm}\forall i,j\in\mathcal{N},\\
    &&\mathrm{C5c:}\hspace*{2mm}y_{m,m',i,j}\geq b_m[i]+b_m[j]-1,\hspace*{2mm}m\neq m',\hspace*{2mm} \forall m,m'\in\mathcal{M},\hspace*{2mm}\forall i,j\in\mathcal{N}\notag,
\end{eqnarray}
where $y_{m,m',i,j}$ is a binary auxiliary variable. For the sake of notation simplicity, we define a binary vector $\mathbf{y}=[y_{1,2,1,1},\cdots,y_{m,m',i,j},\cdots,y_{M-1,M,N,N}]$, $m\neq m'$, $\forall m,m'\in\mathcal{M}$ and $\forall i,j\in\mathcal{N}$ to collect all binary auxiliary variables.
\end{lemma}
% \begin{proof}
%     Please refer to 
% \end{proof}
\par
In addition, we observe that SINR constraint C1 is also non-convex. We note that if an arbitrary $\mathbf{x}_k$ satisfies the constraints in \eqref{Ori_Problem}, %then the value of the objective function remains unchanged when multiplied by an arbitrary phase shift $e^{j\phi}$, resulting in $e^{j\phi}\mathbf{x}_k$ that also satisfies the constraints. 
$\mathbf{x}_k$ multiplied by an arbitrary phase shift $e^{j\phi}$ also satisfies the constraints while the value of the objective function remains unchanged. Thus, we can leverage the following lemma to transform the non-convex SINR constraint in C1 into two equivalent convex constraints, cf. also \cite{luo2006introduction}.
% \begin{equation}\label{LPDC}
%     \overline{\mbox{C2b}}\mbox{:}\hspace*{1mm}\mathrm{Tr}\big(\mathbf{U}\big)-M\leq 0.\vspace*{-2mm}
% \end{equation}
% \begin{eqnarray}
% && \hspace*{-12mm}{{\mbox{C1a}}}\mbox{:}\hspace*{1mm}\sqrt{\sum_{k'\in\mathcal{K}\setminus\{k\}}|\hat{\mathbf{h}}_k^H\mathbf{x}_{k'}|^2+\sigma_k^2}-\frac{\operatorname{Re}\{{\hat{\mathbf{h}}}_{k}^H\mathbf{x}_k\}}{\sqrt{\gamma_k}}\leq 0, \hspace*{1mm}\forall k\in\mathcal{K},\\
% &&\hspace*{-12mm}\mbox{C1b:}\hspace*{1mm}\operatorname{Im}\{{\hat{\mathbf{h}}}_{k}^H\mathbf{x}_k\}=0,\hspace*{1mm}\forall k\in\mathcal{K}.\vspace*{-2mm}
% \end{eqnarray}
\begin{lemma}
Without loss of optimality, we assume that $\hat{\mathbf{h}}_k^H\mathbf{x}_k\in\mathbb{R}$. Then, constraint C1 can be equivalently rewritten as
\vspace*{-2mm}
\begin{eqnarray}
&& \hspace*{-12mm}{{\mathrm{C1a}}}\mbox{:}\hspace*{1mm}\sqrt{\sum_{k'\in\mathcal{K}\setminus\{k\}}|\hat{\mathbf{h}}_k^H\mathbf{x}_{k'}|^2+\sigma_k^2}-\frac{\operatorname{Re}\{{\mathbf{h}}_{k}^H\mathbf{x}_k\}}{\sqrt{\gamma_k}}\leq 0, \hspace*{1mm}\forall k\in\mathcal{K},\\
&&\hspace*{-12mm}\mathrm{C1b:}\hspace*{1mm}\operatorname{Im}\{{\mathbf{h}}_{k}^H\mathbf{x}_k\}=0,\hspace*{1mm}\forall k\in\mathcal{K}.\vspace*{-2mm}
\end{eqnarray}
C1a and C1b are both convex constraints.\vspace*{-2mm}
\end{lemma}
\vspace*{-2mm}
\begin{proof}
Please refer to \cite[Section \rom{3}]{luo2006introduction}.
%Please refer to \cite[Appendix \rom{2}]{zhang2008joint}.\vspace*{-2mm}
\end{proof}
Thus, resource allocation optimization problem \eqref{Ori_Problem} can be equivalently reformulated as follows 
\begin{eqnarray}
\label{Reform_Problem}
    %&&\hspace*{-4mm}\underset{\mathbf{X},\mathbf{W},\bm{\Phi}}{\mino}\hspace*{2mm}\min_{k\in\mathcal{K}}\hspace*{4mm}\mathrm{SINR}_k\notag\\
    &&\hspace*{-4mm}\underset{\mathbf{X},\mathbf{W},\mathbf{B},\mathbf{U},\mathbf{V},\mathbf{y}}{\mino}\hspace*{2mm}\sum_{k\in\mathcal{K}}\left\|\mathbf{w}_k\right\|_2^2\\
    %&&\hspace*{1mm}\mbox{s.t.}\hspace*{7mm} \mbox{C1:}\hspace*{1mm}\sum_{k\in\mathcal{K}}\left\|\mathbf{w}_k\right\|_2^2\leq P,\notag\\
    &&\hspace*{2mm}\mbox{s.t.}\hspace*{12mm} \mbox{C1a},\mbox{C1b}, \mbox{C2a},{\mbox{C2b}},\mbox{C3},\mbox{C4},\mbox{C5a},\mbox{C5b},\mbox{C5c}\notag.
    %&&\hspace*{14mm}\mbox{C3:}\hspace*{1mm} b_m[n]\in \{0,1\},\hspace*{1mm} \forall n\in \mathcal{N}, \forall m\in \mathcal{M},
    %&&\hspace*{14mm}\mbox{C3:}\hspace*{1mm}\theta_n\in\left\{0,\Delta\theta,\cdots,(L-1)\Delta\theta\right\}, \hspace*{1mm} \forall n\in \mathcal{N}.\vspace*{-5mm}
\end{eqnarray}
\begin{remark}
The MINLP problem in \eqref{Reform_Problem} is a convex optimization problem with respect to the continuous variables $\mathbf{X},\mathbf{W},\mathbf{U}$, and $\mathbf{V}$ if the discrete variables $\mathbf{B}$ and $\mathbf{y}$ are fixed. Meanwhile, it is a linear programming problem with respect to the discrete variables $\mathbf{B}$ and $\mathbf{y}$ if the continuous variables $\mathbf{X},\mathbf{W},\mathbf{U}$, and $\mathbf{V}$ are fixed. Hence, the GBD approach is guaranteed to converge to the globally optimal solution of \eqref{Reform_Problem} \cite{geoffrion1972generalized}.
\end{remark}
% \begin{remark}
%     The MINLP problem in \eqref{Reform_Problem} can be recast to a resource allocation problem for multiple AS system that selects $M$ antennas from a unipolar planar array (UPA) equipped with $N$ antennas by increasing the step size $d$ to the antenna spacing of UPA and reducing the minimum distance constraints. 
% \end{remark}
%Next, we develop a GBD-based algorithm to obtain the global optimum of \eqref{Reform_Problem}.%Let $h_{kn}$ denote the channel coefficient between the $n$-th candidate position of the MA elements and user $k$.
\subsection{GBD Procedure}
In order to obtain the globally optimal solution of the MILNP problem in \eqref{Reform_Problem}, we develop an iterative algorithm based on the GBD method. Specifically, the MILNP problem in \eqref{Reform_Problem} is first decomposed into a primal problem and a master problem. In each iteration, we update an upper bound (UB) of the objective function in \eqref{Reform_Problem} by solving the primal problem with the discrete variables $\mathbf{B}$ and $\mathbf{y}$ being fixed. In addition, we solve the master problem by fixing the continuous variables $\mathbf{X},\mathbf{W},\mathbf{U}$, and $\mathbf{V}$ to update a lower bound (LB) of \eqref{Reform_Problem}. In the following, we detail the formulation and solution of the primal and master problems in the $i$-th iteration of the GBD algorithm, followed by an explanation of the overall GBD algorithm.

\subsubsection{Primal Problem}
Using the discrete variables $\mathbf{B}^{(i-1)}$ and $\mathbf{y}^{(i-1)}$ obtained from the master problem in the $(i-1)$-th iteration, the primal problem in the $i$-th iteration is given by
\begin{eqnarray}
\label{Primal_problem}
    %&&\hspace*{-6mm}\underset{\substack{\mathbf{X},\mathbf{W},\mathbf{B},\mathbf{U},\mathbf{V},\\z\in\mathbb{R},y_k\in \mathbb{C}}}{\mino}\hspace*{6mm}z\notag\\
    &&\hspace*{-6mm}\underset{\mathbf{X},\mathbf{W},\mathbf{U},\mathbf{V}}{\mino}\hspace*{2mm}\sum_{k\in\mathcal{K}}\left\|\mathbf{w}_k\right\|_2^2\notag\\
    %&&\hspace*{0mm}\mbox{s.t.}\hspace*{10mm} \mbox{C1:}\hspace*{1mm}\sum_{k\in\mathcal{K}}\left\|\mathbf{w}_k\right\|_2^2\leq P,\notag\\
    &&\hspace*{0mm}\mbox{s.t.}\hspace*{8mm} \mbox{C1a}, \mbox{C1b},{\mbox{C2b}}\\
    &&\hspace*{13mm}\mbox{C2a:}\hspace*{1mm}\begin{bmatrix}
    \mathbf{U} & \mathbf{X} & \mathbf{B}^{(i-1)} \\
    \mathbf{X}^H & \mathbf{V} & \mathbf{W}^H \\
    (\mathbf{B}^{(i-1)})^T & \mathbf{W} & \mathbf{I}_K \\
    \end{bmatrix}\succeq \mathbf{0}.\notag
    %&&\hspace*{13mm}\overline{\mbox{C2b}}:\hspace*{1mm} \mathrm{Tr}\big(\mathbf{U}\big)-\sum_{n=1}^N\widebar{\mathbf{h}}_n^T\mathbf{b}_n^{(i-1)}\leq 0.
    %&&\hspace*{15mm}\mbox{C3:}\hspace*{1mm}\mathrm{Tr}\big(\mathbf{U}\big)-\sum_{n=1}^N\widehat{\mathbf{h}}_n^T\mathbf{b}_n\leq 0,\notag\\
    %&&\hspace*{15mm}\mbox{C3:}\hspace*{1mm}\sum_{l=1}^{L}b_n[l]=1,\ \forall n,\notag\\
    %&&\hspace*{15mm}\mbox{C4:}\hspace*{1mm}b_n[l]\in\left\{0,\hspace*{1mm}1\right\},\hspace*{1mm}\forall l,n\notag\\
     %&&\hspace*{15mm}\mbox{C4:}\hspace*{1mm}\operatorname{Im}\{\hat{\mathbf{h}}_k^H\mathbf{x}_k\}=0,\hspace*{1mm}\forall k\notag\\
    %&&\hspace*{15mm}\mbox{C5:}\hspace*{1mm}
   % 2\operatorname{Re}\{y_k^*\hat{\mathbf{h}}_k^H\mathbf{x}_k\}-|y_k|^2\left(\sum_{k'\in\mathcal{K}\setminus\{k\}}|\hat{\mathbf{h}}_k^H\mathbf{x}_{k'}|^2+\sigma_k^2\right)\geq z,\ \forall k.
   %\vspace*{-6mm}
\end{eqnarray}
The primal problem in \eqref{Primal_problem} is a convex problem, which can be solved by a standard convex programming solver such as CVX \cite{grant2008cvx}. However, the problem \eqref{Primal_problem} is not feasible for all possible discrete variables $\mathbf{B}$ and $\mathbf{y}$. If the primal problem \eqref{Primal_problem} in the $i$-th iteration is feasible, the optimal solution of \eqref{Primal_problem} is denoted by $\mathbf{X}^{(i)}$, $\mathbf{W}^{(i)}$, $\mathbf{U}^{(i)}$, and $\mathbf{V}^{(i)}$ and iteration index $i$ is included in the index set of the feasible iterations $\mathcal{F}$. Next, let $\bm{\Lambda}=\left\{\mu_k,\nu_k,\bm{\Xi},{\xi}\right\}$ denote the collection of Lagrangian multipliers of \eqref{Primal_problem}, where $\mu_k\in\mathbb{R},\nu_k\in\mathbb{R},\bm{\Xi}\in\mathbb{C}^{(N+K+M)\times (N+K+M)}$, and ${\xi}\in\mathbb{R}$ represent the dual variables for constraints C1a, C1b, C2a, and C2b, respectively. The Lagrangian function of the primal problem is given by 
\begin{equation}\label{Largrangian_conf}
    \begin{aligned}
            \mathcal{L}(\mathbf{X},\mathbf{W},\mathbf{U},\mathbf{V},\mathbf{B}^{(i-1)},\bm{\Lambda})&=\sum_{k\in\mathcal{K}}\left\|\mathbf{w}_k\right\|_2^2+f(\mathbf{X},\mathbf{W},\mathbf{U},\mathbf{V},\bm{\Lambda})+2\mathrm{Re}\left\{\mathrm{Tr}\big(\mathbf{B}^{(i-1)}\bm{\Xi}_{31}\big)\right\},
    \end{aligned}\vspace*{-2mm}
\end{equation}
where
\begin{equation}\notag
\begin{aligned}
f(\mathbf{X},\mathbf{W},\mathbf{U},\mathbf{V},\bm{\Lambda})&=\sum_{k\in\mathcal{K}}\mu_k\left(\sqrt{\sum_{k'\in\mathcal{K}\setminus\{k\}}\hspace*{-2mm}|\hat{\mathbf{h}}_k^H\mathbf{x}_{k'}|^2+\sigma_k^2}-\frac{\operatorname{Re}\{\hat{\mathbf{h}}_k^H\mathbf{x}_k\}}{\sqrt{\gamma_k}}\right)\\
    &+\sum_{k\in\mathcal{K}}\nu_k\big(\operatorname{Im}\{\hat{\mathbf{h}}_k^H\mathbf{x}_k\}\big)+\mathrm{Tr}\big(\mathbf{U}\bm{\Xi}_{11}\big)+\mathrm{Tr}\big(\mathbf{V}\bm{\Xi}_{22}\big),\\
    &+2\mathrm{Re}\left\{\mathrm{Tr}\big({\bm{\Xi}^H_{32}}\mathbf{W}\big)+\mathrm{Tr}\big(\mathbf{X}\bm{\Xi}_{21}\big)\right\}.
\end{aligned}
    %f_1(\mathbf{X},\bm{\Lambda})=\sum_{k\in\mathcal{K}}\left[\mu_k\left(\sqrt{\sum_{k'\in\mathcal{K}\setminus\{k\}}|\hat{\mathbf{h}}_k^H\mathbf{x}_{k'}|^2+\sigma_k^2}-\frac{\operatorname{Re}\{\hat{\mathbf{h}}_k^H\mathbf{x}_k\}}{\sqrt{\gamma_k}}\right)+\nu_k\big(\operatorname{Im}\{\hat{\mathbf{h}}_k^H\mathbf{x}_k\}\big)\right],
\end{equation}
Here, the dual matrix $\bm{\Xi}$ is decomposed into nine sub-matrices as follows:
\begin{equation}\label{Qi}
  \bm{\Xi}=\left[ \begin{array}{ccc}
        \bm{\Xi}_{11} & \bm{\Xi}_{21}^H & \bm{\Xi}_{31}^H\\
        \bm{\Xi}_{21} & \bm{\Xi}_{22} & \bm{\Xi}_{32}^H\\
        \bm{\Xi}_{31} & \bm{\Xi}_{32} & \bm{\Xi}_{33}
    \end{array}\right],\vspace*{-2mm}
\end{equation}
where $\bm{\Xi}_{11}\in\mathbb{C}^{N\times N}$, $\bm{\Xi}_{21}\in\mathbb{C}^{K\times N}$, $\bm{\Xi}_{22}\in\mathbb{C}^{K\times K}$, $\bm{\Xi}_{31}\in\mathbb{C}^{M\times N}$, $\bm{\Xi}_{32}\in\mathbb{C}^{M\times K}$, and $\bm{\Xi}_{33}\in\mathbb{C}^{M\times M}$ denote the corresponding sub-matrices of $\bm{\Xi}$. 

 On the other hand, if the primal problem \eqref{Primal_problem} in the $i$-th iteration is not feasible for the given $\mathbf{B}^{(i-1)}$ and $\mathbf{y}^{(i-1)}$, we solve the following feasiblity-check problem:
\begin{eqnarray}
\label{Feasible_problem}
    %&&\hspace*{-6mm}\underset{\substack{\mathbf{X},\mathbf{W},\mathbf{B},\mathbf{U},\mathbf{V},\\z\in\mathbb{R},y_k\in \mathbb{C}}}{\mino}\hspace*{6mm}z\notag\\
    &&\hspace*{-6mm}\underset{\mathbf{X},\mathbf{W},\mathbf{U},\mathbf{V},\bm{\lambda}}{\mino}\hspace*{2mm}\sum_{k\in\mathcal{K}}\lambda_k\notag\\
    %&&\hspace*{0mm}\mbox{s.t.}\hspace*{10mm} \mbox{C1:}\hspace*{1mm}\sum_{k\in\mathcal{K}}\left\|\mathbf{w}_k\right\|_2^2\leq P,\notag\\
    &&\hspace*{0mm}\mbox{s.t.}\hspace*{8mm}\mbox{C1b},\mbox{C2a},{\mbox{C2b}}, \notag\\
 %   &&\hspace*{12mm}\mbox{C1a}\mbox{:}\hspace*{1mm}\sqrt{\sum_{k'\in\mathcal{K}\setminus\{k\}}|\hat{\mathbf{h}}_k^H\mathbf{x}_{k'}|^2+\sigma_k^2}-\frac{\operatorname{Re}\{{\mathbf{h}}_{k}^H\mathbf{x}_k\}}{\sqrt{\gamma_k}}\leq \lambda_k, \hspace*{1mm} \forall k\in \mathcal{K},\notag\\
    &&\hspace*{13mm}\overline{\mbox{C1a}}\mbox{:}\hspace*{1mm}\sqrt{\sum_{k'\neq k}|\hat{\mathbf{h}}_k^H\mathbf{x}_{k'}|^2+\sigma_k^2}-\frac{\operatorname{Re}\{{\mathbf{h}}_{k}^H\mathbf{x}_k\}}{\sqrt{\gamma_k}}\leq \lambda_k,\\
    &&\hspace*{21mm}\forall k\in \mathcal{K},\notag\\
    &&\hspace*{13mm}\mbox{C6}\mbox{:}\hspace*{1mm}\lambda_k\geq 0, \forall k\in\mathcal{K}\notag,\vspace*{-2mm}
    %&&\hspace*{15mm}\mbox{C3:}\hspace*{1mm}\mathrm{Tr}\big(\mathbf{U}\big)-\sum_{n=1}^N\widehat{\mathbf{h}}_n^T\mathbf{b}_n\leq 0,\notag\\
    %&&\hspace*{15mm}\mbox{C3:}\hspace*{1mm}\sum_{l=1}^{L}b_n[l]=1,\ \forall n,\notag\\
    %&&\hspace*{15mm}\mbox{C4:}\hspace*{1mm}b_n[l]\in\left\{0,\hspace*{1mm}1\right\},\hspace*{1mm}\forall l,n\notag\\
     %&&\hspace*{15mm}\mbox{C4:}\hspace*{1mm}\operatorname{Im}\{\hat{\mathbf{h}}_k^H\mathbf{x}_k\}=0,\hspace*{1mm}\forall k\notag\\
    %&&\hspace*{15mm}\mbox{C5:}\hspace*{1mm}
   % 2\operatorname{Re}\{y_k^*\hat{\mathbf{h}}_k^H\mathbf{x}_k\}-|y_k|^2\left(\sum_{k'\in\mathcal{K}\setminus\{k\}}|\hat{\mathbf{h}}_k^H\mathbf{x}_{k'}|^2+\sigma_k^2\right)\geq z,\ \forall k.
\end{eqnarray}
where $\bm{\lambda}=[\lambda_1,\cdots,\lambda_K]$ denotes an auxiliary optimization variable. Note that \eqref{Feasible_problem} is always feasible and convex. Thus, we can solve \eqref{Feasible_problem} with a standard CVX solver. Similar to \eqref{Primal_problem}, the optimal solution of \eqref{Feasible_problem} is denoted by $\widetilde{\mathbf{X}}^{(i)},\widetilde{\mathbf{W}}^{(i)},\widetilde{\mathbf{U}}^{(i)}$, and $\widetilde{\mathbf{V}}^{(i)}$. Then, the iteration index $i$ is included in the index set of the infeasible iterations $\mathcal{I}$. Furthermore, we formulate the Lagrangian function of \eqref{Feasible_problem} as 
\begin{equation}\label{Largrangian_feasible}
\begin{aligned}
    \widetilde{\mathcal{L}}(\mathbf{X},\mathbf{W},\mathbf{U},\mathbf{V},\mathbf{B}^{(i-1)},\widetilde{\bm{\Lambda}})&={f}(\mathbf{X},\mathbf{W},\mathbf{U},\mathbf{V},\widetilde{\bm{\Lambda}})\\
    &+2\mathrm{Re}\left\{\mathrm{Tr}\big(\mathbf{B}^{(i-1)}\Tilde{\bm{\Xi}}_{31}\big)\right\},
\end{aligned}\vspace*{-2mm}
\end{equation}
where $\widetilde{\bm{\Lambda}}=[\widetilde{\mu}_k,\widetilde{\nu_k},\widetilde{\bm{\Xi}},\widetilde{\xi}]$ is the collection of the dual variables $\widetilde{\mu}_k,\widetilde{\nu_k},\widetilde{\bm{\Xi}}\in\mathbb{C}^{(N+K+M)\times (N+K+M)}$, and $\widetilde{\xi}$ for constraints $\overline{\mbox{C1a}}$, C1b, C2a, and C2b, respectively. Similar to the notation in \eqref{Qi}, $\widetilde{\bm{\Xi}}_{31}$ denotes the submatrix of $\widetilde{\bm{\Xi}}$ formed by rows $\{N+K+1,\cdots,N+K+M\}$ and columns $\{1,\cdots,N\}$. The solution of the feasibility-check problem \eqref{Feasible_problem} is used to separate the infeasible solutions $\mathbf{B}^{(i-1)}$ and $\mathbf{y}^{(i-1)}$ from the feasible set of the master problem in the following iterations.
\subsubsection{Master Problem}
The master problem is formulated based on nonlinear convex duality theory \cite{geoffrion1972generalized}. Without loss of generality, we recast the master problem into the following epigraph form by introducing auxiliary optimization variable $\eta$:
\begin{eqnarray}
\label{Master_problem}
    %&&\hspace*{-6mm}\underset{\substack{\mathbf{X},\mathbf{W},\mathbf{B},\mathbf{U},\mathbf{V},\\z\in\mathbb{R},y_k\in \mathbb{C}}}{\mino}\hspace*{6mm}z\notag\\
    &&\hspace*{-8mm}\underset{\mathbf{B},\mathbf{y},\eta}{\mino}\hspace*{2mm}\eta\notag\\
    %&&\hspace*{0mm}\mbox{s.t.}\hspace*{10mm} \mbox{C1:}\hspace*{1mm}\sum_{k\in\mathcal{K}}\left\|\mathbf{w}_k\right\|_2^2\leq P,\notag\\
    &&\hspace*{-4mm}\mbox{s.t.}\hspace*{10mm}\mbox{C3}, \mbox{C4},\mbox{C5a}, \mbox{C5b},\mbox{C5c}\notag\\
    &&\hspace*{10mm}\mbox{C7a}\mbox{:}\hspace*{1mm}\eta\geq\min_{\substack{\mathbf{X},\mathbf{W},\\\mathbf{U},\mathbf{V}}}\mathcal{L}(\mathbf{X},\mathbf{W},\mathbf{U},\mathbf{V},\mathbf{B},\bm{\Lambda}^{(t)}),\hspace*{0mm}\forall t\in\{1,\cdots,i\}\cap\mathcal{F},\\
    &&\hspace*{10mm}\mbox{C7b}\mbox{:}\hspace*{1mm}0\geq\min_{\substack{\widetilde{\mathbf{X}},\widetilde{\mathbf{W}},\\\widetilde{\mathbf{U}},\widetilde{\mathbf{V}}}}\widetilde{\mathcal{L}}(\widetilde{\mathbf{X}},\widetilde{\mathbf{W}},\widetilde{\mathbf{U}},\widetilde{\mathbf{V}},\mathbf{B},\widetilde{\bm{\Lambda}}^{(t)}),\hspace*{0mm}\forall t\in\{1,\cdots,i\}\cap\mathcal{I}\notag,\vspace*{-2mm}
    %&&\hspace*{15mm}\mbox{C3:}\hspace*{1mm}\mathrm{Tr}\big(\mathbf{U}\big)-\sum_{n=1}^N\widehat{\mathbf{h}}_n^T\mathbf{b}_n\leq 0,\notag\\
    %&&\hspace*{15mm}\mbox{C3:}\hspace*{1mm}\sum_{l=1}^{L}b_n[l]=1,\ \forall n,\notag\\
    %&&\hspace*{15mm}\mbox{C4:}\hspace*{1mm}b_n[l]\in\left\{0,\hspace*{1mm}1\right\},\hspace*{1mm}\forall l,n\notag\\
     %&&\hspace*{15mm}\mbox{C4:}\hspace*{1mm}\operatorname{Im}\{\hat{\mathbf{h}}_k^H\mathbf{x}_k\}=0,\hspace*{1mm}\forall k\notag\\
    %&&\hspace*{15mm}\mbox{C5:}\hspace*{1mm}
   % 2\operatorname{Re}\{y_k^*\hat{\mathbf{h}}_k^H\mathbf{x}_k\}-|y_k|^2\left(\sum_{k'\in\mathcal{K}\setminus\{k\}}|\hat{\mathbf{h}}_k^H\mathbf{x}_{k'}|^2+\sigma_k^2\right)\geq z,\ \forall k.
\end{eqnarray}
where $\bm{\Lambda}^{(t)}$ and $\widetilde{\bm{\Lambda}}^{(t)}$ denote the collections of the optimal Lagrangian multipliers for feasible primal problem \eqref{Primal_problem} and feasibility-check problem \eqref{Feasible_problem} in the $t$-th iteration, respectively. Constraints C7a and C7b correspond to the optimality cut and feasibility cut, respectively \cite{geoffrion1972generalized}. The inner minimization in C7a and C7b can be obtained from the optimal solutions of the primal problem \eqref{Primal_problem} and the feasibility-check problem \eqref{Feasible_problem} exploiting the following lemma:
\begin{lemma}
Inequality constraints C7a and C7b can be recast as the following two linear inequalities:
\begin{eqnarray}
\label{Recast_C5}
    %&&\hspace*{-6mm}\underset{\substack{\mathbf{X},\mathbf{W},\mathbf{B},\mathbf{U},\mathbf{V},\\z\in\mathbb{R},y_k\in \mathbb{C}}}{\mino}\hspace*{6mm}z\notag\\
   % &&\hspace*{-8mm}\underset{\mathbf{B},\eta}{\mino}\hspace*{2mm}\eta\notag\\
    %&&\hspace*{0mm}\mbox{s.t.}\hspace*{10mm} \mbox{C1:}\hspace*{1mm}\sum_{k\in\mathcal{K}}\left\|\mathbf{w}_k\right\|_2^2\leq P,\notag\\
    %&&\hspace*{-4mm}\mbox{s.t.}\hspace*{2mm}\mbox{C3a}, \mbox{C3b},\notag\\
    &&\hspace*{2mm}\overline{\mathrm{C7a}}\mbox{:}\hspace*{1mm}\eta\geq \sum_{k\in\mathcal{K}}\left\|\mathbf{w}_k^{(t)}\right\|_2^2+f(\mathbf{X}^{(t)},\mathbf{W}^{(t)},\mathbf{U}^{(t)},\mathbf{V}^{(t)},\bm{\Lambda}^{(t)})+2\mathrm{Re}\left\{\mathrm{Tr}\big(\mathbf{B}\bm{\Xi}_{31}\big)\right\},\notag\\
    &&\hspace*{12mm}\forall t\in\{1,\cdots,i\}\cap\mathcal{F},\\
    &&\hspace*{2mm}\overline{\mathrm{C7b}}\mbox{:}\hspace*{1mm}0\geq f(\mathbf{X}^{(t)},\mathbf{W}^{(t)},\mathbf{U}^{(t)},\mathbf{V}^{(t)},\bm{\Lambda}^{(t)})+2\mathrm{Re}\left\{\mathrm{Tr}\big(\mathbf{B}\bm{\Xi}_{31}\big)\right\},\hspace*{0mm}\forall  t\in\{1,\cdots,i\}\cap\mathcal{I},\vspace*{-2mm}
    %&&\hspace*{15mm}\mbox{C3:}\hspace*{1mm}\mathrm{Tr}\big(\mathbf{U}\big)-\sum_{n=1}^N\widehat{\mathbf{h}}_n^T\mathbf{b}_n\leq 0,\notag\\
    %&&\hspace*{15mm}\mbox{C3:}\hspace*{1mm}\sum_{l=1}^{L}b_n[l]=1,\ \forall n,\notag\\
    %&&\hspace*{15mm}\mbox{C4:}\hspace*{1mm}b_n[l]\in\left\{0,\hspace*{1mm}1\right\},\hspace*{1mm}\forall l,n\notag\\
     %&&\hspace*{15mm}\mbox{C4:}\hspace*{1mm}\operatorname{Im}\{\hat{\mathbf{h}}_k^H\mathbf{x}_k\}=0,\hspace*{1mm}\forall k\notag\\
    %&&\hspace*{15mm}\mbox{C5:}\hspace*{1mm}
   % 2\operatorname{Re}\{y_k^*\hat{\mathbf{h}}_k^H\mathbf{x}_k\}-|y_k|^2\left(\sum_{k'\in\mathcal{K}\setminus\{k\}}|\hat{\mathbf{h}}_k^H\mathbf{x}_{k'}|^2+\sigma_k^2\right)\geq z,\ \forall k.
\end{eqnarray}
respectively.%, where $\widetilde{\bm{\Xi}}_{31}^{(i)}$ denotes the submatrix of $\widetilde{\bm{\Xi}}^{(i)}$ formed by rows $\{N+K+1,\cdots,N+K+M\}$ and columns $\{1,\cdots,N\}$.
\end{lemma}
\begin{proof}
    We first study the inner minimization problem in C7a for feasible iteration index $t$: 
    \begin{eqnarray}
            &&\hspace*{4mm}\min_{{\mathbf{X},\mathbf{W},\mathbf{U},\mathbf{V}}}\mathcal{L}(\mathbf{X},\mathbf{W},\mathbf{U},\mathbf{V},\mathbf{B},\bm{\Lambda}^{(t)})\notag\\
            &&=\min_{{\mathbf{X},\mathbf{W},\mathbf{U},\mathbf{V}}}\left\{\sum_{k\in\mathcal{K}}\left\|\mathbf{w}_k\right\|_2^2+f(\mathbf{X},\mathbf{W},\mathbf{U},\mathbf{V},\bm{\Lambda}^{(t)})\right\}+2\mathrm{Re}\left\{\mathrm{Tr}\big(\mathbf{B}\bm{\Xi}_{31}\big)\right\}\notag\\
            &&\overset{(a)}{=}\sum_{k\in\mathcal{K}}\left\|\mathbf{w}_k^{(t)}\right\|_2^2+f(\mathbf{X}^{(t)},\mathbf{W}^{(t)},\mathbf{U}^{(t)},\mathbf{V}^{(t)},\bm{\Lambda}^{(t)})+2\mathrm{Re}\left\{\mathrm{Tr}\big(\mathbf{B}\bm{\Xi}_{31}^{(t)}\big)\right\},
    \end{eqnarray}
    where equality (a) holds due to the optimality condition of the Lagrangian function for a convex optimization problem. Similarly, we can also prove that 
    \begin{eqnarray}
    &&\hspace*{4mm}\min_{\substack{\widetilde{\mathbf{X}},\widetilde{\mathbf{W}},\\\widetilde{\mathbf{U}},\widetilde{\mathbf{V}}}}\widetilde{\mathcal{L}}(\widetilde{\mathbf{X}},\widetilde{\mathbf{W}},\widetilde{\mathbf{U}},\widetilde{\mathbf{V}},\mathbf{B},\widetilde{\bm{\Lambda}}^{(t)})\notag\\&&=f(\mathbf{X}^{(t)},\mathbf{W}^{(t)},\mathbf{U}^{(t)},\mathbf{V}^{(t)},\bm{\Lambda}^{(t)})+2\mathrm{Re}\left\{\mathrm{Tr}\big(\mathbf{B}\Tilde{\bm{\Xi}}_{31}^{(t)}\big)\right\},
    \end{eqnarray}
    if the primal problem \eqref{Primal_problem} is infeasible in the $t$-th iteration.
\end{proof}
Note that by replacing inequality constraints C7a and C7b by constraints $\overline{\mbox{C7a}}$ and $\overline{\mbox{C7b}}$, respectively, optimization problem \eqref{Master_problem} is recast as a mixed integer linear programming (MILP) problem, which can be solved by employing standard numerical solvers for MILPs, e.g.,  MOSEK \cite{grant2008cvx}. The optimal solution of \eqref{Master_problem} in the $i$-th iteration is denoted as $\mathbf{B}^{(i)}$ and $\mathbf{y}^{(i)}$.
\subsubsection{Overall Algorithm}
The overall procedure of the proposed optimal algorithm is summarized in \textbf{Algorithm 1}. Before the first iteration, we set index $i$ to zero and initialize $\mathbf{B}^{(0)}$ and $\mathbf{y}^{(0)}$ to a feasible solution. In the $i$-th iteration, we begin by solving problem \eqref{Primal_problem}. If the problem is feasible, we generate the optimality cut for the master problem in \eqref{Master_problem} based on the intermediate solution $\mathbf{X}^{(i)}$, $\mathbf{W}^{(i)}$, $\mathbf{U}^{(i)}$, $\mathbf{V}^{(i)}$, and their corresponding Lagrangian multiplier set $\bm{\Lambda}^{(i)}$. Furthermore, we update the upper bound $\mathrm{UB}^{(i)}$ of \eqref{Reform_Problem} with the objective value $\sum_{k\in\mathcal{K}}\left\|\mathbf{w}_k^{(i)}\right\|_2^2$ obtained in the current iteration. If problem \eqref{Primal_problem} is infeasible, we turn to solve the feasibility-check problem in \eqref{Feasible_problem} to generate the feasibility cut for the master problem. Subsequently, we optimally solve the master problem in \eqref{Master_problem} using a standard MILP solver. The objective value of the master problem provides a performance lower bound $\mathrm{LB}^{(i)}$ for the original optimization problem in \eqref{Ori_Problem}. By following this procedure, we iteratively reduce the gap between the $\mathrm{LB}$ and $\mathrm{UB}$ in each iteration. Based on \cite[Theorem 2.4]{geoffrion1972generalized}, the proposed GBD-based algorithm is guaranteed to converge to the globally optimal solution of \eqref{Reform_Problem} in a finite number of iterations for a given convergence tolerance $\Delta\geq 0$. Although the worst case computational complexity of the proposed GBD-based algorithm scales exponentially with the number of MA elements, in our simulation experiments, the proposed GBD method converged in significantly fewer iterations than an exhaustive search.
\begin{algorithm}[t]
\caption{Optimal Resource Allocation Algorithm}
\begin{algorithmic}[1]
\small
\STATE Set iteration index $i=0$, initialize upper bound $\mathrm{UB}^{(0)}\gg 1$, lower bound $\mathrm{LB}^{(0)}=0$, the set of feasible iterations indices $\mathcal{F}=\emptyset$, the set of infeasible iterations indices $\mathcal{I}=\emptyset$, and convergence tolerance $\Delta\ll 1$, generate a feasible $\mathbf{B}^{(0)}$.
\REPEAT
\STATE Set $i=i+1$
\STATE Solve \eqref{Primal_problem} for given $\mathbf{B}^{(i-1)}$, and $\mathbf{y}^{(i-1)}$.
\IF{the primal problem \eqref{Primal_problem} is feasible}
\STATE Update $\mathbf{X}^{(i)},\mathbf{W}^{(i)},\mathbf{U}^{(i)}$, and $\mathbf{V}^{(i)}$ and store the corresponding objective function value of $\sum_{k\in\mathcal{K}}\left\|\mathbf{w}_k^{(i)}\right\|_2^2$
\STATE Construct $\mathcal{L}(\mathbf{X},\mathbf{W},\mathbf{U},\mathbf{V},\mathbf{B},\bm{\Lambda}^{(i)})$ based on \eqref{Largrangian_conf}
\STATE Update the upper bound of \eqref{Reform_Problem} as $\mathrm{UB}^{(i)}=\min\left\{\mathrm{UB}^{(i-1)},\hspace*{1mm}\sum_{k\in\mathcal{K}}\left\|\mathbf{w}_k^{(i)}\right\|_2^2\right\}$, and update $\mathcal{F}$ by $\mathcal{F}\cup\{i\}$ %and $\mathcal{I}^{(i)}=\mathcal{I}^{(i-1)}$
\ELSE
\STATE Solve \eqref{Feasible_problem}, update $\widetilde{\mathbf{X}}^{(i)}$, $\widetilde{\mathbf{W}}^{(i)}$, $\widetilde{\mathbf{U}}^{(i)}$, $\widetilde{\mathbf{V}}^{(i)}$
\STATE Construct $\widetilde{\mathcal{L}}(\mathbf{X},\mathbf{W},\mathbf{U},\mathbf{V},\mathbf{B},\widetilde{\bm{\Lambda}}^{(i)})$ based on \eqref{Largrangian_feasible}
\STATE Update $\mathcal{I}$ by $\mathcal{I}\cup\{i\}$ %and $\mathcal{F}=\mathcal{F}$
\ENDIF
\STATE Solve the relaxed master problem \eqref{Master_problem} and update $\eta^{(i)}$, $\mathbf{B}^{(i)}$, and $\mathbf{y}^{(i)}$.
\STATE Update the lower bound as $\mathrm{LB}^{(i)}=\eta^{(i)}$
\UNTIL $\mathrm{UB}^{(i)}-\mathrm{LB}^{(i)}\leq \Delta$
\end{algorithmic}
\end{algorithm}

\section{Numerical Results}
In this section, we evaluate the performance of the proposed optimal algorithm via numerical simulations. We consider a system where the BS is equipped with $M=4$ MA elements to provide communication service for $K=4$ single-antenna users. The carrier frequency is set to $5$ GHz, i.e., the wavelength is $\lambda=0.06$ m. The transmitter area is a square area of size $l\lambda\times l\lambda$, where $l$ is the normalized transmitter area size at the BS. Due to the properties of the MA driver, the transmitter area is quantized into discrete positions with equal distance $d$ as shown in Fig. 1. The minimum distance $D_{\mathrm{min}}$ is set to $0.015$ m. The users are randomly distributed, and their distance to the BS is uniformly distributed between $20$ m to $100$ m. The noise variance of each user is set to $-80$ dBm, $\forall k \in\mathcal{K}$. As in \cite{ma2022mimo,zhu2022modeling,zhu2023movable}, the channel coefficient $h_{m,k}(\mathbf{p}_n)$ between the $m$-th MA element and the $k$-th user at $\mathbf{p}_n$ is modeled as follows\footnote{The adopted field-response channel model in \cite{ma2022mimo,zhu2022modeling,zhu2023movable} leverages the amplitude, phase, and angle of arrival/departure information on each multipath component under far-field condition to characterize the general multipath channel for MA-enabled systems.},
\begin{equation}
    h_{m,k}(\mathbf{p}_n)=\mathbf{1}_{L_p}^T\bm{\Sigma}_k\mathbf{g}_k(\mathbf{p}_n),
\end{equation}
where $\mathbf{1}_{L_p}$ is the all-one field response vector (FRV) at the $k$-th user equipped with a single fixed-antenna. Diagonal matrix $\bm{\Sigma}_k=\mathrm{diag}\{[\sigma_{1,k},\cdots,\sigma_{L_p,k}]^T\}$ contains the path responses of all $L_p=16$ channel paths from the transmitter area to the $k$-th user. All path response coefficients $\sigma_{l_p,k}, \forall l_p\in\{1,\cdots,L_p\}$ are independently and identically distributed and follow complex Gaussian distribution $\mathcal{CN}(0,L_0D_k^{-\alpha})$, where $L_0$, $D_k$, and $\alpha=2.2$ denote the large-scale fading at reference distance $d_0=1$ m, the distance from BS to the $k$-th user, and the path loss exponent, respectively. $\mathbf{g}_k(\mathbf{p}_n)$ denotes the transmit FRV between the $k$-th user and the MA at position $\mathbf{p}_n$, which is given by \cite{ma2022mimo,zhu2022modeling}
\begin{equation}
    \mathbf{g}_k(\mathbf{p}_n)=\left[e^{j\rho_{k,1}(\mathbf{p}_n)},\cdots,e^{j\rho_{k,L_p}(\mathbf{p}_n)}\right]^T,
\end{equation}
where $\rho_{k,l_p}(\mathbf{p}_n)=\frac{2\pi}{\lambda}\left((x_k-x_1)\cos\theta_{k,l_p}\sin\phi_{k,l_p}+(y_k-y_1)\sin\theta_{k,l_p}\right)$ represents the phase difference of the $l_p$-th channel path between $\mathbf{p}_n$ and the first position $\mathbf{p}_1$. $\theta_{k,l_p}$ and $\phi_{k,l_p}$ denote the
elevation and azimuth angles of departure of the $l_p$-th channel paths for the $k$-th user, respectively, and follow the probability density function $f_{\mathrm{AoD}}(\theta_{k,l_p}, \phi_{k,l_p})=\frac{\cos\theta_{k,l_p}}{2\pi}$, $\theta_{k,l_p}\in [- \pi/2,\pi/2]$, $\phi_{k,l_p}\in [- \pi/2,\pi/2]$. 
%The effective holographic MISO channel between the BS and the $k$-th user $\hat{\mathbf{h}}_k$ is considered to be Rician distributed and modeled as follows:
% \begin{equation}\vspace*{-1mm}
%     \hat{\mathbf{h}}_k=\sqrt{L_0D_k^{-\alpha}}\left(\sqrt{\frac{\beta}{1+\beta}}\mathbf{h}^{\mathrm{L}}_k+\sqrt{\frac{1}{1+\beta}}\mathbf{h}^{\mathrm{N}}_k\right),
% \end{equation}%\vspace*{-1mm}
%where $L_0$ denotes the large-scale fading at reference distance $d_0=1$ m. $\alpha=2.2$ and $\beta=1$ denote the path loss exponent and the Rician factor, respectively. Vectors $\mathbf{h}^{\mathrm{L}}_k$ and $\mathbf{h}^{\mathrm{N}}_k$ denote the LoS component and the NLoS component, respectively, where $\mathbf{h}^{\mathrm{L}}_k$ is modeled by the array response vector of the effective holographic MIMO surface while $\mathbf{h}^{\mathrm{N}}_k$ is generated by following Rayleigh distribution \cite{wu2023globally},\cite{yu2021robust}. 
In this work, we set the step size $d$ as $0.01$ m for both the proposed scheme and the baseline schemes, as per the experimental setup in \cite{zhuravlev2015experimental}. In addition, we consider also a step size $d$ of $\lambda/2=0.03$ m to investigate the effect of step size on system performance. %In fact, for $d=\lambda/2$, the corresponding multi-antenna optimization problem is equivalent to the problem of selecting multiple antennas in a $(2l+1)\times (2l+1)$ uniform planar array (UPA) with antenna spacing $\lambda/2$, which serves as the optimal performance of AS in the given transmit area.

We consider three baseline schemes for comparison. For baseline scheme 1, the MA elements are fixed at $M$ positions that satisfy the minimum distance constraint and are chosen randomly from $\mathcal{P}$. The beamforming vectors $\mathbf{w}_k$ are obtained by solving the beamforming problem for fixed antenna arrays employing semidefinite relaxation \cite{xu2020resource}. For baseline scheme 2, we adopt the AS technique, where the BS is equipped with a $2\times M$ uniform planar array (UPA) with fixed-position antennas spaced by $\lambda/2=0.03$ m such that the channels corresponding to different antennas are statistically independent. We solve the beamforming problem for all possible subsets of $M$ antenna elements and selected the optimal subset that minimizes the BS transmit power. Note that the antenna spacing cannot be adjusted in baseline scheme 2. For baseline scheme 3, we design a suboptimal iterative algorithm based on AO. Specifically, in each iteration, beamforming matrix $\mathbf{W}$ is optimized for a fixed binary decision matrix $\mathbf{B}$ obtained in the last iteration. Afterwards, $\mathbf{B}$ is updated in a block-coordinated-descent (BCD) manner for the fixed $\mathbf{W}$ obtained in the current iteration. 

\begin{figure}
    \centering
    \includegraphics[width=2.6in]{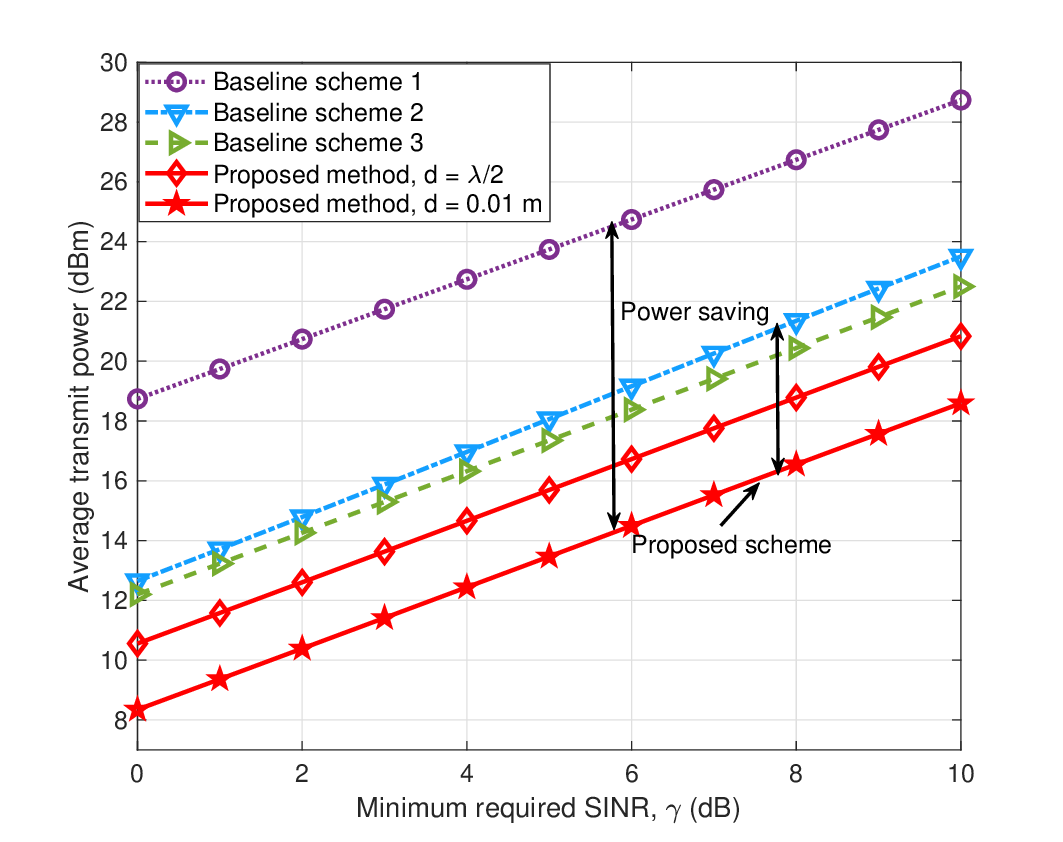}
    \caption{Average BS transmit power versus the minimum required SINR of the users.}
    \label{fig:SINR}
\end{figure}

Fig. \ref{fig:SINR} shows the average BS transmit power required for the considered schemes versus the users' minimum required SINR values $\gamma_k=\gamma$, $\forall k$, where the transmitter area of the MA at the BS is $2\lambda\times 2\lambda=0.12$ m $\times$ $0.12$ m. For the considered SINR range and $d=0.01$ m, the proposed scheme requires on average approximately $120$ iterations to obtain the global optimal solution, which is much faster than an exhaustive search over all possible positions of the $M=4$ MAs. It is observed that as the minimum required SINR value increases, the BS consumes more transmit power to satisfy the more rigorous quality-of-service requirements of the users. Furthermore, we also observe that the proposed scheme outperforms the three baseline schemes for the entire range of $\gamma$. In particular, for baseline scheme 1, the BS is equipped with an antenna array with fixed antenna positions, i.e., the spatial correlation of the transmit antenna array is not optimal. Although baseline scheme 2 employs AS to increase the spatial DoFs at the BS, which leads to a performance improvement compared to baseline scheme 1, the spatial resolution ($\lambda/2$) is quite coarse due to the fixed-position antenna setting. As for baseline scheme 3, the adopted AO algorithm optimizes the positions of the MA elements and the beamforming matrix, leading to a $5$ dB gain compared to baseline scheme 1. However, as the AO-based algorithm is generally suboptimal due to its local search, it may get trapped in stationary points, resulting in a $4$ dB performance gap compared to the proposed optimal solution,
%limited performance improvement over the iterations. On the contrary, by jointly designing the MA positions and beamforming matrix, the transmit BS power of the proposed scheme can be further reduced by $5$ dB,
 which highlights the significance of fully exploiting the DoFs provided by the proposed MA-enabled system. Moreover, increasing the step size of the electromechanical device to $d=\lambda/2$ leads to a performance loss of roughly $2$ dB, indicating the trade-off between the transmit power and the MA control precision in MA-enabled systems.
 % Furthermore, we can refer to the proposed scheme with $d=\lambda/2$ as the optimal AS for a $9\times 9$ UPA with spacing $\lambda/2$. Thus, the MA-enabled system can achieve a slight gain with a significantly less number of antenna elements compared to the conventional MIMO with AS technique. 

\begin{figure}
    \centering
    \includegraphics[width=2.6in]{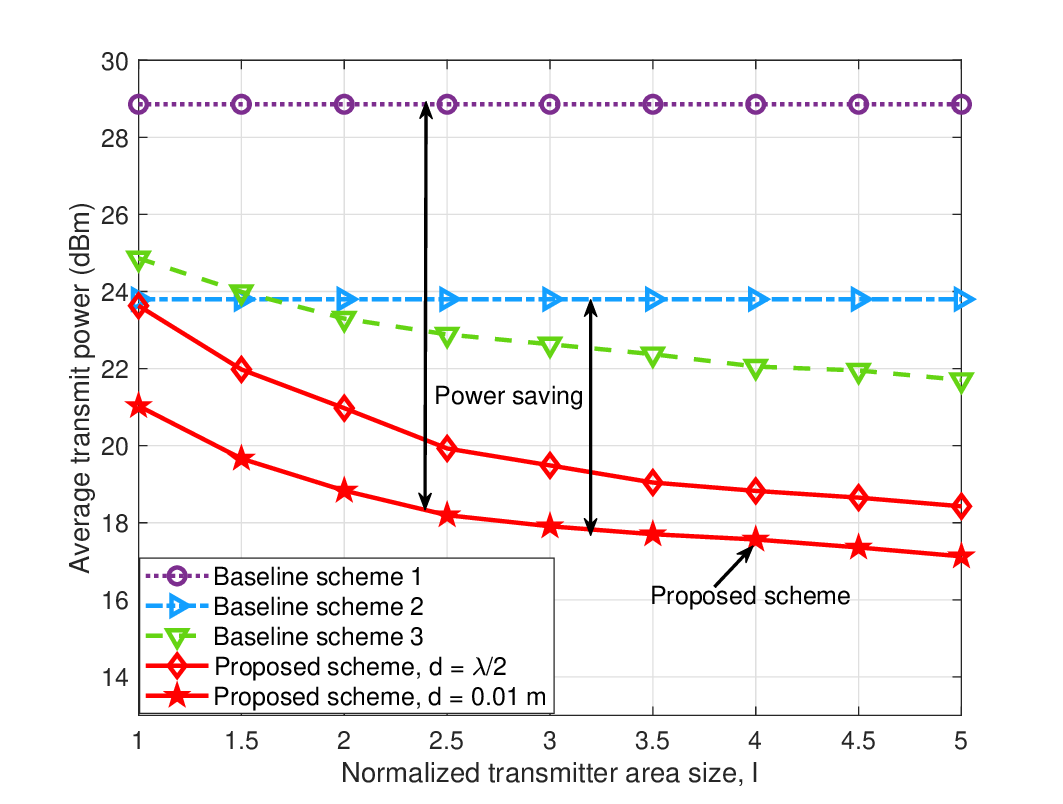}
    \caption{Average BS transmit power versus the normalized transmitter area size of the MA-enabled system.}
    \label{fig::Region}
\end{figure}

Fig. \ref{fig::Region} depicts the relationship between the BS transmit power and the normalized transmitter area size for $\gamma=10$ dB. We observe that for the proposed scheme and the AO-based scheme, the BS transmit power decreases as the normalized transmitter area size increases. This can be attributed to the fact that a larger transmitter area allows higher flexibility for the positioning of the MA elements for shaping the desired spatial correlation, leading to a potential performance gain. On the other hand, baseline schemes 1 and 2 with their fixed antenna positions cannot benefit from a larger transmitter area, indicating that baseline schemes 1 and 2 cannot fully utilize the spatial DoFs provided by a larger transmitter area. It is worth noting that the proposed optimal scheme outperforms the three baseline schemes in terms of the BS transmit power. Specifically, for larger transmitter areas, the gap between the AO-based scheme and the optimal design is enlarged since the AO approach is more likely to converge to a local optimum. Furthermore, the performance of the proposed scheme saturates when the normalized transmitter area size reaches
$3$, indicating that the performance of MA-enabled systems cannot be further improved by using even larger transmitter areas. %Moreover, increasing the step size results in a $2$ dB additional transmit BS power for a very small transmit region ($l=1.5$) as feasible MA positions with favorable channels hardly exist in the small feasible sets caused by the small transmit antenna region and large step size.

%Furthermore, we can observe that increasing the step size leads to a $2$ dB performance loss for a very small transmit region $l=1.5$. In fact, for a very small transmit region, reducing the step size can dramatically increase the feasible set size of the MA positions, leading to a higher possibility of finding the MA positions with good channels. 
\section{Conclusion}
In this work, we investigated for the first time the optimal resource allocation design for a multiuser MA-enabled downlink MISO communication system. Due to the practical hardware limits, we modeled the movement of the MA elements as a discrete motion. Then, we formulated an optimization problem for the minimization of the BS transmit power while guaranteeing a minimum SINR of the users. We showed that the globally optimal solution of the formulated optimization problem could be obtained with an iterative GBD-based algorithm. Our simulation results revealed the superiority of the proposed multiuser MA-enabled MISO system and the optimality of the proposed GBD-based algorithm. In future work, we will develop a computationally efficient suboptimal scheme for a practical MA-enabled system with imperfect CSI.
\bibliographystyle{IEEEtran}
\bibliography{reference.bib}
\end{document}